\documentclass[a4paper]{panle}
\usepackage{url}
\usepackage{cite}
\usepackage{wrapfig}
\usepackage{graphicx}
\usepackage{amssymb}
\usepackage{amsfonts}
\usepackage{amsmath}
\usepackage{longtable}
\usepackage{rotating}
\usepackage{lscape}
\usepackage{epsfig}
\usepackage{multirow}
\usepackage{bm}
\usepackage{amsthm}
\usepackage{physics}

\usepackage{lineno}

\newtheorem{lemma}{Lemma}
\newtheorem{theorem}{Theorem}

\originalTeX
\begin{document}

\title{Quantum Hashing with QKD States}
\maketitle
\authors{A.\,V.~Vasiliev\,$^{a,b,}$\footnote{E-mail: vav.kpfu@gmail.com}, I.\,G.~Zinnatullin\,$^{a,b,}$\footnote{E-mail: IlnGZinnatullin@kpfu.ru}}
\setcounter{footnote}{0}
\from{$^{a}$\,Kazan Federal University, Kazan, Russian Federation}
\from{$^{b}$\,Zavoisky Physical-Technical Institute, FRC Kazan Scientific Center of RAS, Kazan, Russian Federation
}


\begin{abstract}

Quantum hashing is a well-known technique. In this paper, we present a new construction of a quantum hash function based on binary error-correcting codes. The construction of the quantum hash strongly resembles the quantum state preparation in QKD protocols (e.g., BB84 protocol) and can be implemented on current QKD-devices.
\end{abstract}
\vspace*{6pt}

\label{sec:intro}
\section*{INTRODUCTION}

Quantum hashing is the technique that allows for verification of the hashed value without disclosing it. In previous papers, we presented different versions of quantum hash functions and gave provable bounds for their cryptographic resistance (for example, see \cite{AKVZ:2025:Theory-of-QHashing}). 

However, the physical implementation of the most efficient quantum hash functions requires advanced quantum engineering and might be relevant for the future quantum devices beyond NISQ era.

Therefore, in this paper, we present a new quantum hash function that can be fully implemented on the basis of existing hardware for the quantum key distribution (QKD) protocol. In this new proposal, we use only the states from the original BB84 protocol \cite{BB84}, and thus quantum procedures for generating and verifying quantum hash already exist in commercial hardware. So, the update required to implement quantum hashing by such systems is in their classical part.

\label{sec:preparation}
\section{QUANTUM HASHING BASED ON ECC}

We call a binary linear code \(C\) of length \(m\), dimension \(n\), and minimum distance \(d\) an \([m, n, d]\)-code. This code can be represented as
\[
    C = \{w(x) = (w^{x}_{1}, \ldots, w^{x}_{m}) \in \{0, 1\}^{m} \;|\; x \in \{0, 1\}^{n}\},
\]
where \(w(x)\) denotes the codeword corresponding to \(x\).

Let \(k\) be the number of qubits. Let \(C\) be the \([m, n, d]\)-code with \(m = 2k\) and \(G\) be its generator matrix. For \(x \in \{0, 1\}^{n}\) we compute its codeword
\[
    w(x) = xG = (w^{x}_{1}, \ldots, w^{x}_{m}) = (v^{x}_{1}, \ldots, v^{x}_{k}, b^{x}_{1}, \ldots, b^{x}_{k}) \in \{0, 1\}^{m}
\]
and define its quantum hash composed of \(k\) single qubits as follows:
\[
    \ket{\psi_{C}(x)} = \bigotimes_{i=1}^{k} \ket{\psi^{i}_{C}(x)}
\]
where
\[
    \ket{\psi^{i}_{C}(x)} = H^{b^{x}_{i}}X^{v^{x}_{i}}\ket{0} =
    \begin{cases}
        \ket{0},& \text{if } v^{x}_{i} = 0 \text{ and } b^{x}_{i} = 0,\\
        \ket{1},& \text{if } v^{x}_{i} = 1 \text{ and } b^{x}_{i} = 0,\\
        \ket{+} = \frac{\ket{0} + \ket{1}}{\sqrt{2}},& \text{if } v^{x}_{i} = 0 \text{ and } b^{x}_{i} = 1,\\
        \ket{-} = \frac{\ket{0} - \ket{1}}{\sqrt{2}},& \text{if } v^{x}_{i} = 1 \text{ and } b^{x}_{i} = 1,\\
    \end{cases}
\]
for \(1 \leq i \leq k\). Here \(X\) is the Pauli-\(X\) gate, \(H\) is the Hadamard gate. Thus, the second half of the codeword \(w(x)\) is used to select the basis, whereas the first half is used to select the quantum state (basis vector) in the chosen basis.

We define the collision threshold for the function \(\psi_{C}\) by
\[
    F(\psi_{C}) = \max_{x \neq y} F(\ket{\psi_{C}(x)}, \ket{\psi_{C}(y)}),
\]
where \(F\) denotes the fidelity defined as \(F(\ket{\psi_{C}(x)}, \ket{\psi_{C}(y)}) = |\braket{\psi_{C}(x)}{\psi_{C}(y)}|^{2}\).

\begin{lemma}\label{lemma:collision-resistance}
Let \(C\) be the \([m, n, d]\)-code with \(m = 2k\). The collision bound of the quantum hash function \(\psi_{C}\) is the following:
\[
    F(\psi_{C}) \leq \left( \frac{1}{2} \right)^{\lceil \frac{d}{2} \rceil}.
\]
\end{lemma}
\begin{proof}
Let \(x, y \in \{0, 1\}^{n}\) with \(x \neq y\). By construction of the quantum hash
\[
    \braket{\psi_{C}(x)}{\psi_{C}(y)} = \prod_{i=1}^{k} \braket{\psi^{i}_{C}(x)}{\psi^{i}_{C}(y)}.
\]
The corresponding codewords are
\[
    w(x) = xG = (v^{x}_{1},\ldots, v^{x}_{k},b^{x}_{1},\ldots, b^{x}_{k}),
\]
\[
    w(y) = yG = (v^{y}_{1},\ldots, v^{y}_{k},b^{y}_{1},\ldots, b^{y}_{k}).
\]
Since \(x \neq y\) and the code \(C\) has the minimum distance \(d\), we have
\[
    d_{H}(w(x), w(y)) \geq d,
\]
where \(d_{H}\) denotes the Hamming distance. Consider the pairs \((v^{x}_{i}, b^{x}_{i})\) and \((v^{y}_{i}, b^{y}_{i})\) for \(1 \leq i \leq k\). If \((v^{x}_{i}, b^{x}_{i}) \neq (v^{y}_{i}, b^{y}_{i})\), we deal with the following three cases:
\begin{enumerate}
    \item \(v^{x}_{i} \neq v^{y}_{i}, b^{x}_{i} = b^{y}_{i} \Rightarrow \left|\braket{\psi^{i}_{C}(x)}{\psi^{i}_{C}(y)}\right| = 0\),
    \item \(v^{x}_{i} = v^{y}_{i}, b^{x}_{i} \neq b^{y}_{i} \Rightarrow \left|\braket{\psi^{i}_{C}(x)}{\psi^{i}_{C}(y)}\right| = \frac{1}{\sqrt{2}}\),
    \item \(v^{x}_{i} \neq v^{y}_{i}, b^{x}_{i} \neq b^{y}_{i} \Rightarrow \left|\braket{\psi^{i}_{C}(x)}{\psi^{i}_{C}(y)}\right| = \frac{1}{\sqrt{2}}\).
\end{enumerate}
Therefore, if \((v^{x}_{i}, b^{x}_{i}) \neq (v^{y}_{i}, b^{y}_{i})\), then
\[
    \left|\braket{\psi^{i}_{C}(x)}{\psi^{i}_{C}(y)}\right|^{2} \leq \frac{1}{2}.
\]

Let \(\mathcal{I} = \{i \;|\; (v^{x}_{i}, b^{x}_{i}) \neq (v^{y}_{i}, b^{y}_{i}), 1 \leq i \leq k\}\). Note that each index \(i \in \mathcal{I}\) contributes at most \(2\) to the \(d_{H}(w(x), w(y))\). Hence
\[
    d \leq d_{H}(w(x),w(y)) \leq 2|\mathcal{I}|,
\]
which implies \(|\mathcal{I}| \geq \lceil d/2 \rceil\).

Consequently, we have
\[
    \left|\braket{\psi_{C}(x)}{\psi_{C}(y)}\right|^{2} = \prod_{i=1}^{k} \left|\braket{\psi^{i}_{c}(x)}{\psi^{i}_{C}(y)}\right|^{2} \leq \left(\frac{1}{2}\right)^{|\mathcal{I}|} \leq \left(\frac{1}{2}\right)^{\lceil \frac{d}{2} \rceil}.
\]

Finally, we take the maximum over all distinct inputs and obtain
\[
    F(\psi_{C}) \leq \left( \frac{1}{2} \right)^{\lceil \frac{d}{2} \rceil}.
\]
\end{proof}

\section{Quantum Hashing Based on BCH Codes}

In this subsection, we construct an \([m, n, d]\)-code \(C\) with \(d \geq 2\lceil \log n \rceil + 1\) and \(m = 2k\), where \(k = \lceil n/2 \rceil + O(\log^{2} n)\). The lower bound on \(d\) yields the collision bound \(F(\psi_{C}) \leq 1/(2n)\). The asymptotic bound on \(k\) yields the pre-image resistance of \(\psi_{C}\).

Let \(n \geq 17\). We set \(l = \lceil \log n \rceil \geq 5\) and \(t = l + 1\). Consider a binary primitive narrow-sense  BCH code \(C_{\mathrm{BCH}}\) of length \(M = 2^{t} - 1\) and designed distance \(\delta = 2l + 1\). We denote the dimension of \(C_{\mathrm{BCH}}\) by \(N\) and its minimum distance by \(D\). It is known that the following bounds hold for \(C_{\mathrm{BCH}}\) \cite[Theorem 1, page 258]{MacWilliams-Sloane:1978:The-Theory-of-Error-Correcting-Codes}:
\[
    N \geq M - tl,\quad D \geq \delta.
\]

We can verify that \(N > n\). Since \(l \geq 5\), we have \(l(l+1) + 1 < 2^{l}\). Therefore,
\[
    N \geq M - tl = 2^{l+1} - (l(l+1) + 1) > 2^{l} \geq n.
\]

Let \(\mathcal{I} \subseteq \{1, \ldots, M\}\) be an information set of \(C_{\mathrm{BCH}}\), that is, a set of \(N\) coordinate positions that carry the information. We select \(\mathcal{J} \subseteq \mathcal{I}\) with \(|\mathcal{J}| = N - n > 0\) and shorten \(C_{\mathrm{BCH}}\) in the coordinates indexed by \(\mathcal{J}\):
\[
    C^{\prime} = \{(w_{i})_{i \notin \mathcal{J}} \;|\; w \in C_{\mathrm{BCH}}, w_{j} = 0 \text{ for all } j \in \mathcal{J}\}.
\]
Since \(\mathcal{J} \subseteq \mathcal{I}\), the shortening reduces both the length and the dimension by \(|\mathcal{J}|\) and does not decrease the minimum distance. Hence \(C^{\prime}\) has dimension \(n^{\prime} = N - (N - n) = n\), length
\[
    m^{\prime} = M - (N - n) = n + (M - N) \leq n + tl,
\]
and minimum distance \(d^{\prime} \geq D \geq 2l + 1\).

It remains to ensure that the code length is even. If \(m^{\prime}\) is even, then we set \(C = C^{\prime}\). If \(m^{\prime}\) is odd, then we obtain \(C\) by appending an additional zero coordinate to every codeword of \(C^{\prime}\). This operation preserves both the dimension and the minimum distance. Consequently, \(C\) is an \([m, n, d]\)-code with
\[
    m = 2 \left\lceil\frac{m^{\prime}}{2}\right\rceil,\quad d = d^{\prime} \geq 2l + 1.
\]

Since \(m = 2k\) and \(m' \geq n\), we have
\begin{equation}\label{eq:condition_k_value}
    \left\lceil \frac{n}{2} \right\rceil
    \leq k = \left\lceil \frac{m^{\prime}}{2} \right\rceil
    \leq \left\lceil \frac{n+l(l+1)}{2} \right\rceil
    \leq \left\lceil \frac{n}{2} \right\rceil + \frac{l(l+1)}{2}.
\end{equation}
Since \(l = \lceil \log n \rceil\), Eq.~\eqref{eq:condition_k_value} implies
\[
    k = \left\lceil \frac{n}{2} \right\rceil + O(\log^{2} n).
\]

Furthermore,
\begin{equation}\label{eq:condition_l_value}
    \left(\frac{1}{2}\right)^{l} = \left(\frac{1}{2}\right)^{\lceil \log n \rceil} \leq \left(\frac{1}{2}\right)^{\log n} \leq \frac{1}{n}.
\end{equation}

The bound on \(d\) gives
\begin{equation}\label{eq:condition_d_value}
    \left\lceil \frac{d}{2} \right\rceil \geq l + 1.
\end{equation}

By Lemma~\ref{lemma:collision-resistance} and Eqs.~\eqref{eq:condition_d_value} and \eqref{eq:condition_l_value} we obtain
\[
    F(\psi_{C}) \leq \left(\frac{1}{2}\right)^{\lceil \frac{d}{2} \rceil} \leq \left(\frac{1}{2}\right)^{l + 1} \leq \frac{1}{2n}.
\]

To ensure \(k < n\), it suffices by Eq.~\eqref{eq:condition_k_value} to require
\begin{equation}\label{eq:k_less_n_condition_1}
   \left\lceil \frac{n+l(l+1)}{2} \right\rceil  \leq n-1.
\end{equation}
In fact Eq.~\eqref{eq:k_less_n_condition_1} is equivalent to
\begin{equation}\label{eq:k_less_n_condition_2}
    n \geq l(l+1) + 2.
\end{equation}
Eq.~\eqref{eq:k_less_n_condition_2} holds for \(n = 32\) and \(n \geq 44\).

\begin{theorem}
Let \(n = 32\) or \(n \geq 44\) and \(l = \lceil \log n \rceil\). There exists an \([m, n, d]\)-code \(C\) with \(m = 2k\), \(d \geq 2\lceil \log n \rceil + 1\), and
\[
    \left\lceil \frac{n}{2} \right\rceil \leq k \leq \left\lceil \frac{n}{2} \right\rceil + \frac{l(l+1)}{2} < n.
\]
This code \(C\) gives rise to the quantum hash function \(\psi_{C}\) that maps \(n\)-bit inputs into \(k\)-qubit outputs with
\[
    k = \left\lceil\frac{n}{2}\right\rceil + O(\log^{2} n) < n,\quad F(\psi_{C}) \leq \frac{1}{2n}.
\]
\end{theorem}

\bigskip

\section*{CONCLUSIONS}

In this work we present a new construction of a quantum hash function based on error-correcting codes. We show that there exists a code that allows to ensure the desired cryptographic properies. We also note that the implementation of the corresponding states is relatively easy on current quantum hardware.

\section*{FUNDING}
The research has been supported by Russian Science Foundation Grant 25-11-00366.


\section*{CONFLICT OF INTEREST}
The authors of this work declare that they have no conflicts of interest.

\bibliographystyle{pepan}
\bibliography{references}

\end{document}